\documentclass[11pt]{article}
\usepackage{setspace}

\usepackage{graphicx} 
\usepackage{amsmath,amssymb,amsthm}
\usepackage{color}
\usepackage[margin=1in]{geometry}
\usepackage{tikz}
\usetikzlibrary{arrows.meta, positioning}
\usepackage{pdflscape}
\usepackage[square,numbers]{natbib}
\usepackage{hyperref}
\usepackage{enumitem}
\usepackage{pgfplots}
\pgfplotsset{compat=newest}
\usepgfplotslibrary{fillbetween}

\newtheorem{theorem}{Theorem}
\newtheorem{proposition}[theorem]{Proposition}	
\newtheorem{corollary}[theorem]{Corollary}	
\newtheorem{assumption}[theorem]{Assumption}	
\newtheorem{definition}[theorem]{Definition}
\newtheorem{example}[theorem]{Example}
\newtheorem{remark}{Remark}

\numberwithin{equation}{section}	
\numberwithin{theorem}{section}

\newcommand{\R}{\mathbb{R}}

\usepackage[normalem]{ulem}

\title{Designing On-Chain Options: Amortizing Perpetual Options}
\author{
Maxim Bichuch\thanks{Department of Mathematics, SUNY at Buffalo, Buffalo, NY 14260, {\tt mbichuch@buffalo.edu}.}
\and 
Zachary Feinstein\thanks{School of Business, Stevens Institute of Technology, Hoboken, NJ 07030, {\tt  zfeinste@stevens.edu}. Portions of this work are the subject of a provisional U.S.\ patent application filed by the author.}
}
\date{\today}

\begin{document}

\maketitle
\begin{abstract}
Financial options are fundamental to traditional markets, enabling strategies ranging from hedging to speculating. Yet, while the Automated Market Maker paradigm has revolutionized decentralized spot markets, no equivalent standard has emerged for on-chain options. Typical designs attempt to replicate centralized exchange mechanics, requiring high-frequency oracles and robust liquidation engines which may fail during stress events. 
This paper presents a design for amortizing perpetual options tailored to the operational and adversarial constraints of blockchain environments. Leveraging this primitive, we introduce a decentralized market framework with minimal consistency requirements. We demonstrate that this contract functions as a foundational risk primitive for DeFi, enabling applications such as endogenous collateralization and explicitly priced de-peg insurance, thereby showing that this design provides a layer for mutualizing tail risk across protocols without reliance on centralized clearing institutions.
\end{abstract}

\section{Introduction}\label{sec:intro}
Financial options have long been a key component to the global financial system.
Despite this ubiquity in traditional financial markets, options have not yet been widely used in decentralized finance (DeFi); instead, for cryptocurrencies, the most liquid derivatives are perpetual futures~\cite[Table 1]{pennella2025unified}. However, though perpetual futures allow for leverage and synthetic shorting, the linear payoff of these instruments preclude them from being used for, e.g., effective hedging of volatility risks.
Furthermore, for the options that do exist, nearly 99\% are traded in centralized exchanges,\footnote{\url{https://insights.derive.xyz/announcing-lyra-v2/}} with Deribit's self-reported 85\% market share in BTC and ETH options.\footnote{\url{https://www.deribit.com/}} This is in contrast to spot markets in which decentralized exchanges regularly handle 15-20\% of all volume.\footnote{\url{https://www.theblock.co/data/decentralized-finance/dex-non-custodial/dex-to-cex-spot-trade-volume}} 

Given this divergence between the success of linear markets (i.e., spot trading and perpetual futures) and the failures of options on-chain, we refer to~\cite{fateh2025option,pennella2025unified} for recent surveys of the DeFi derivatives landscape. As noted by those works, there have been a number of approaches taken in practice to construct such markets including both peer-to-peer and peer-to-pool models.
Notably, as highlighted in \cite[Table 1]{fateh2025option} and \cite[Table 9]{pennella2025unified}, the vast majority of current protocols attempt to replicate the design paradigms of centralized exchanges. These protocols rely on high-frequency oracles and liquidation mechanisms to increase capital efficiency. This approach, however, conflicts with the latency constraints and limited liveness assumptions of decentralized blockchains, introducing attack vectors for manipulation and MEV. To implement these systems, often these protocols require app-chains or partially off-chain systems. 

However, an early proposed use of AMMs was to replicate \emph{concave} payoffs as the eponymous replicating market makers (RMMs)~\cite{angeris2023replicating}. As specifically designed AMMs, the option holder buys the position by providing liquidity to the RMM which is rebalanced to the appropriate payoff by arbitrageurs. In exchange for taking this concave payoff, the option holder is compensated in the form of transaction fees.
The valuation of AMM liquidity positions as derivatives has been studied in, e.g.,~\cite{bichuch2024defi}.
In particular, the concentrated liquidity positions of Uniswap v3~\cite{adams2021uniswap} can be viewed as approximations of \emph{short} (perpetual) put options with physical delivery; the tighter the position is concentrated around the strike price, the better this approximation of the payoff. This concept was implemented by Panoptic~\cite{lambert2022panoptic} which allows users to short these liquidity positions to effectively replicate long perpetual options. As with RMMs, the buyer needs to compensate the underwriter based on the fees in the underlying liquidity position to hold these options. Such explicit payments require robust margin accounts and liquidation of the position if these required installment payments are not made in full. Furthermore, because of the differences in underlying liquidity positions, these options are often considered only semi-fungible where a more concentrated and less concentrated option are \emph{not} directly exchangeable.

Within this work, we study the design of option contracts that are compatible with the operational and adversarial constraints of blockchains. We show that many existing approaches fail because they cannot simultaneously satisfy the requirements of autarkic settlement, asynchronous execution, fungibility, composability, and maintain a positive cost-of-carry. We demonstrate that \emph{amortizing perpetual options}, originally proposed in~\cite{feinstein2025amortizing}, resolve this design conflict within the space of non-expiring contracts (Section~\ref{sec:options}). Building on this primitive, we introduce a decentralized market design with minimal consistency requirements inspired by lending pools and prediction markets (Section~\ref{sec:design}). Finally, we demonstrate that this option contract can function as a general risk primitive for DeFi, e.g., enabling endogenous collateralization in lending markets and creating explicitly priced de-peg insurance for stablecoins and liquid staking tokens. In doing so, this primitive can act as a mutualization layer for tail risk across protocols without introducing centralized clearing institutions (Section~\ref{sec:apps}). Section~\ref{sec:conclusion} concludes.

\section{DeFi Option Design}\label{sec:options}

The translation of option contracts to a blockchain based setting requires a fundamental re-evaluation of the instrument's settlement logic. In traditional finance, option design is largely driven by the need for specialized position hedges. In DeFi, the design is constrained by the operational and security limitations of the blockchain. 
Within this section we consider five fundamental constraints on DeFi option design in Section~\ref{sec:options-axioms} and then construct a specific \emph{perpetual} option design satisfying all desired axioms in Section~\ref{sec:options-ampo}.

\begin{table}
\centering
\begin{tabular}{|p{0.17\textwidth}|p{0.17\textwidth}|p{0.6\textwidth}|}
\hline
\textbf{Axiom} & \textbf{Design \newline Requirement} & \textbf{Failure Mode if Violated} \\ \hline\hline
\textbf{Autarkic \newline Settlement} & \emph{Physical \newline Settlement} & Manipulable settlement due to reliance on external pricing signals (MEV/OEV). \\ \hline
\textbf{Asynchronous \newline Execution} & \emph{American \newline Option} & Gas competition, liveness failures, and failed settlement during execution windows.\\ \hline
\textbf{Fungible} & \emph{Exchange \newline Tradable} & Contracts that cannot be pooled, preventing shared liquidity and making scalable DeFi markets infeasible. \\ \hline
\textbf{Composable} & \emph{Self-Custodied} & Account-based or custodial positions that cannot be integrated across protocols. \\ \hline
\textbf{Positive \newline Cost-of-Carry} & \emph{Theta Decay or \newline Funding Rate} & Unbounded time value, eliminating economic incentives for liquidity providers to underwrite the convex risk. \\ \hline\hline
\textbf{Perpetual \newline Lifespan}$^*$ & \emph{No Expiry} & Downstream protocols must import complex, maturity-date logic. \\ \hline
\end{tabular}
\caption{Summary of axioms for decentralized options and their associated failure modes.\\ $^*$Perpetual lifespan is a design extension that improves composability with DeFi, but is not a requirement for the derivative construction.}
\label{table:axioms}
\end{table}

\subsection{Axiomatic Characterization}\label{sec:options-axioms}
Options on digital assets, whether centralized or decentralized, have predominantly followed the contract design of traditional derivatives. While small variations are introduced, e.g., instituting Asian options for cash settlement, the prevailing design principle involves trying to fit novel market mechanics around familiar option contracts. Herein, to establish a foundation for a natively on-chain option, we abstract away from specific legacy contracts and instead approach the design from first principles. Specifically, we consider how option designs must interact with the blockchain; we highlight that blockchains are characterized by asynchronous execution, trust-minimized environments, and permissionless integration. The blockchain, therefore, requires options that are engineered to withstand adversarial conditions while remaining economically sound for all market participants. Formally, these constraints impose a set of five fundamental axioms that any viable DeFi option must satisfy. Any contract that violates one or more of these axioms is either insecure, non-functional, or economically unsustainable in a blockchain environment.

\subsubsection{Autarkic Settlement}\label{sec:options-axioms-autarkic}
To avoid manipulation risks, option contracts should be self-contained and not rely on external smart contracts or oracles to determine the value to be exchanged at exercise. Because financial options are characterized by their \emph{convex} payoffs, which embed a natural leverage, these instruments are highly attractive targets for manipulation through, e.g., maximal extractable value (MEV) or oracle extractable value (OEV). To mitigate these vulnerabilities, prevailing markets for contracts on digital assets deploy Asian designs (e.g., Deribit and Derive); Asian contracts attempt to increase the cost of manipulations by referencing a time-weighted average price over, e.g., 30 or 60 minute intervals. However, this heuristic still relies on an external pricing signal which is only reliable for contracts on liquidly traded underlyings; when applied to the long tail of digital assets, with fragmented liquidity and low volumes, this averaging is still susceptible to manipulation. To eliminate any external manipulation risks, the option design must be fully autarkic.

While traditional markets follow a cash-settled paradigm in which the payoff is computed in some num\'eraire, this construction inherently requires an external oracle to compute the exercise value (e.g., $(S-K)^+$ for a call option). \textbf{Physical settlement} circumvents this requirement entirely. By mandating that the exerciser, e.g., deliver the strike asset in exchange for the underlying (for calls) or the underlying asset in exchange for the strike (for puts), the contract logic remains strictly internal. Because the market does not need to verify the current spot price of the underlying asset to process the exercise settlement, external oracle manipulations become structurally irrelevant to the settlement mechanics of the option. Instead, the holder will exercise if, and only if, he or she believes that doing so is individually advantageous whether or not the quoted external spot price supports that decision.

\subsubsection{Asynchronous Execution}\label{sec:options-axioms-async}
As noted above, the timing of the exercise event introduces operational bottlenecks unique to blockchains. Blockchains process transactions sequentially within discrete blocks, making them highly susceptible to network congestion and transaction reordering. Legacy contracts with \emph{European}-style expiry dictate a specific time that the contract must be exercised (if at all). In a blockchain context, this fixed maturity creates deterministic \emph{event cliffs}: all rational in-the-money actors must race to execute their transactions within an extremely narrow block window. This concentration of activity leads to gas wars, localized liveness failures, and targeted MEV attacks as arbitrageurs compete for block space to process exercise events.

To maintain operational resilience, the option contract must align with the asynchronous reality of the blockchain. \textbf{American options}, which permit the holder to exercise at any point prior to expiry, relax the temporal constraints on exercise. By granting the holder the flexibility to choose the block in which to exercise, American-style options can diffuse settlement activity across the lifespan of the contract. This continuous exercise window eliminates the single point of failure associated with deterministic maturities, ensuring that the network can process the exercise without creating systemic congestion or novel MEV attacks.

\subsubsection{Fungible}\label{sec:options-axioms-fungible}
In traditional financial markets, derivatives broadly divide into two categories: over-the-counter (OTC) contracts and exchange-traded instruments. OTC options allow for customized parameters, but they inherently lack secondary market liquidity because each contract is unique to the counterparties involved. Conversely, \textbf{exchange-traded options} standardize the contract parameters; in this way, these instruments are fungible by construction. This fungibility permits deep market liquidity, as it allows positions to be easily exchanged, netted, and pooled.

To be tradable in the standard DeFi peer-to-pool market model, the instrument must be fungible. As the nomenclature indicates, decentralized liquidity is structured around pooled capital; this is seen in, e.g., augmented bonding curves~\cite{zargham2020curved} for primary issuance or automated market makers (AMMs)~\cite{angeris2020improved,angeris2023geometry,bichuch2022axioms} for secondary trading. These pooled liquidity models mandate that the underlying positions are perfectly interchangeable; a unified pool cannot manage or price a collection of bespoke contracts. Therefore, for an on-chain option to support deep, pooled liquidity, it must strictly adhere to the exchange-tradable paradigm. That is, every unit of notional must be identical and interchangeable.

\subsubsection{Composable}\label{sec:options-axioms-composable}
In traditional finance, composability is achieved through centralized integration. Derivative positions are held within clearinghouses or prime brokerages, allowing them to be cross-margined alongside other assets within that specific institutional silo. However, this model relies on the institution maintaining custody of the assets and acting as the universal ledger. In contrast, decentralized finance operates across a fragmented ecosystem of permissionless, independent protocols without a central clearinghouse.
To achieve composability across this open architecture, an option contract must enable \textbf{self-custody}. This derivative position must be representable as a standardized, transferable token rather than a siloed account balance. This self-custodied token ensures that the option can be held directly in a user's wallet and integrated into distinct protocols, e.g., as collateral within a lending market (see Section~\ref{sec:apps-lending}).

\subsubsection{Positive Cost-of-Carry}\label{sec:options-axioms-coc}
At their core, financial options are instruments that transfer risk between counterparties: the holder secures a convex payoff, while the underwriter assumes the market exposure. Furthermore, to underwrite this contract, the liquidity provider must commit capital as collateral, thereby incurring an opportunity cost. For dated contracts, the underwriter is compensated both upfront by the option premium and continuously via the option's \textbf{theta decay}.

However, perpetual options introduce their own design constraint. A perpetual American option with no holding cost grants the holder indefinite optionality. This introduces a system in which no underwriter would lock up their liquidity for a ``reasonable'' cost, and no buyer would pay the ``fair'' (possibly infinite) price. To maintain finite premiums and attract sustainable liquidity, the instrument must enforce a positive cost-of-carry either implicitly (e.g., via the option theta) or explicitly (e.g., via a funding rate) to penalize indefinite holding.

\subsection{Decentralized Perpetual Options}\label{sec:options-dpo}
Reviewing the axioms established in Section~\ref{sec:options-axioms}, we note that \emph{dated, physically-settled American options} fulfill all design requirements. However, while dated options satisfy these fundamental requirements, the deterministic maturities introduce frictions when imposed on-chain. Within this section, we focus specifically on perpetual option designs. In doing so, we discuss why a perpetual lifespan is a practical necessity for DeFi composability (Section~\ref{sec:options-dpo-comp}), explore the failures of legacy perpetual contracts when measured against our axioms (Sections~\ref{sec:options-dpo-pao} and~\ref{sec:options-dpo-ci}), and finally introduce the \textbf{amortizing perpetual option} as a contract that satisfies all five axioms without imposing a maturity date.

\subsubsection{Composability Superiority of Perpetual Designs}\label{sec:options-dpo-comp}
While dated options can be tokenized and self-custodied, they introduce operational frictions when integrated into downstream DeFi protocols. Though a strict definition of composability simply requires tokenization of positions, a core principle of composability is that these instruments can be embedded in third-party protocols as a modular component, i.e., as collateral in a lending market. If a dated contract is used in this manner, the receiving protocol must explicitly account for the contract's expiration. In particular, because a dated option's value can collapse to zero at maturity, the downstream protocol is forced to implement liquidation engines that are time-aware to force users to either roll or close their positions prior to the expiry. Perpetual options resolve this friction entirely; by removing the maturity date, the option behaves operationally like a standard token. This allows the instrument to be natively integrated into any third-party protocol without requiring it to inherit the maturity logic.

\begin{remark}\label{rem:perpetual-liquidity}
Perpetual designs additionally resolve the structural liquidity fragmentation inherent to dated instruments. For dated options, liquidity is distributed across strikes and expiration dates; replicating this fragmentation in DeFi divides liquidity across isolated smart contracts, resulting in shallow pools and high slippage. A perpetual contract collapses the temporal dimension, consolidating all trading volume and underwriting capital for a given strike into a single liquidity pool. Furthermore, this design eliminates the persistent gas costs associated with actively rolling expiring positions.
\end{remark}

\subsubsection{Perpetual American Options}\label{sec:options-dpo-pao}
Given the aforementioned desire to move from dated, physically-settled American options to a contract with a perpetual lifespan, the na\"{i}ve first approach is to simply consider (physically-settled) perpetual American options (see, e.g.,~\cite[Chapter 26.2]{hull}). However, because this instrument never expires, it inherently lacks theta decay; without this natural time decay, the contract fails to have a cost-of-carry. 
That is, perpetual American options grant the holder indefinite optionality without any ongoing maintenance costs.
As discussed in Section~\ref{sec:options-axioms-coc}, with \emph{zero} cost-of-carry, the market equilibrium results in excessive premia required entice any underwriters, but simultaneously deterring buyers.
Therefore, while perpetual American options satisfy the operational constraints of the blockchain, they violate the positive cost-of-carry axiom, making them economically unviable.

\subsubsection{Perpetual Continuous-Installment (CI) Options}\label{sec:options-dpo-ci}
While standard perpetual American options (see, e.g.,~\cite[Chapter 26.2]{hull}) fail to satisfy the requisite cost-of-carry, perpetual continuous-installment (CI) options introduce an explicit funding rate $c_t$ to be paid from the holder to the underwriter throughout time to maintain optionality~\cite{ciurlia2009note,kimura2009american,kimura2010valuing}. That is, the buyer pays an initial premium $V_0 > 0$ to the underwriter for a claimable notional of $N_0 > 0$ upon exercise; to keep the CI contract alive, the buyer needs to pay continuous installments $c_t dt$ per unit of notional to the underwriter at all times $t$. Under this structure, the holder possesses two forms of optionality:
\begin{enumerate}
\item the right to exercise, and
\item the right to lapse by halting the installment payments.
\end{enumerate}
However, implementing this lapsing logic creates design issues in DeFi. 
Under traditional lapsing (in which the full notional $N_0$ is sustained, exercised, or lapses as a single unit), each holder's contract is \emph{non-fungible}; we note that, where CI options trade, they are over-the-counter instruments. 
Since the full notional lapsing causes this non-fungibility, we want to consider a CI-variant that permits \emph{partial} lapsing. Specifically, any fractional deficit in funding causes an equivalent fraction of the notional to lapse. However, due to the asynchronous execution of the blockchain, such a design would require account-based margining systems (similar to those implemented for perpetual futures) rather than a self-custodied token; as such these CI-variants would be \emph{non-composable}.

\subsubsection{Amortizing Perpetual Options}\label{sec:options-ampo}
To recover composability and fungibility simultaneously, and thus satisfy all five design requirements for blockchain-based options, herein we consider the amortizing perpetual option (AmPO) design of~\cite{feinstein2025amortizing}. AmPOs are a variant of CI options with implicit, rather than explicit, installment payments. That is, instead of directly paying the continuous installments, the option holder ``sells'' a fraction of his or her option holdings back to the underwriter at the current premium over time. Thus, the realized notional follows an exponential decay
\begin{equation}\label{eq:notional}
dN_t = -q_t N_t dt, \quad q_t = c_t / V_t
\end{equation}
for claimable notional $N_t > 0$, premium $V_t > 0$, and installment rate $c_t > 0$ at time $t$. 
In this way, rather than a physical transfer of payments from option holders to underwriters, AmPO holders compensate the underwriter indirectly via the decrease in his or her risk exposure.

\begin{definition}\label{defn:ampo}\cite[Definition 2.2]{feinstein2025amortizing}
An \textbf{amortizing perpetual option (AmPO)} is an implicit payment, perpetual American CI option with installment costs $c_t = q V_t$ for some progressively measurable amortization rate process $(q_t)_{t \geq 0}$ such that $q_t \geq 0$ a.s.\ and $\int_0^t q_s ds < \infty$ a.s.\ for every time $t \geq 0$. That is, for an AmPO with payoff function $\Pi: \R_+ \to \R_+$ on underlying $(S_t)_{t \geq 0}$ purchased at time $t_0$, the realized exercise value is $e^{-\int_{t_0}^\tau q_s ds} \Pi(S_\tau)$ with exercise time $\tau$.
\end{definition}

\begin{remark}
Fair valuation of call and put AmPOs goes beyond the scope of this work. We refer the interested reader to~\cite{feinstein2025amortizing} which proves that AmPOs can be priced as (vanilla) perpetual American options with both the risk-free rate and dividend rate augmented upward by the amortization rate.
\end{remark}

Notably, while Definition~\ref{defn:ampo} highlights that this amortizing option variant is amenable to any (suitably regular) installment scheme,~\cite[Section 4.1]{feinstein2025amortizing} proposed $c_t = q V_t$ as a solution to potential manipulations, e.g., by underwriters forcibly amortizing away their counterparties; this design choice ensures the option holder is subject to a fixed, predictable decay in notional $N_t$ over time via the amortization rate $q > 0$, while the underwriter is compensated directly in proportion to the fair-value exposure they carry. 

\begin{assumption}
For the remainder of this work, consider AmPOs with constant amortization rates $q_t \equiv q > 0$.
\end{assumption}

As noted above, we consider a physically-settled variant of call and put AmPOs. That is, for a call option with strike price $K > 0$ being exercised at time $\tau \geq 0$ after its inception, a user must deliver $K N_0 e^{-q\tau}$ units of cash (i.e., a stablecoin or some other num\'eraire) in order to receive $N_0 e^{-q\tau}$ shares of the underlying asset. Similarly, for a put option with the same strike and amortization rate, a user must deliver $N_0 e^{-q\tau}$ shares of the underlying asset for $K N_0 e^{-q\tau}$ units of cash. Because these exercise events do not rely on any external signals, and follow American-style timing, these physically-settled AmPOs are automatically autarkic and immune to oracle manipulations.

Consider now how to encode these AmPOs as fungible tokens on-chain. This token contract would need immutable parameters encoding the options series: type (call/put); strike ($K > 0$); and amortization rate ($q > 0$). In addition, any primary market for such options has two fundamental (mutable) state variables: the open interest \emph{notional} ($X \geq 0$) and the total market collateral ($C \geq 0$). For ease of discussion, herein we require that the market collateral must match the asset to be delivered, i.e., the underlying for calls and the num\'eraire for puts.
Notably, for put options, although the collateral is posted in the num\'eraire, the open-interest notional is defined in units of the underlying asset; accordingly, for accounting purposes, we implicitly convert the num\'eraire collateral $\bar{C}$ into notional-equivalent units via $C := \bar{C}/K$.
Furthermore, due to the amortization of assets, without any market operations, the open interest naturally follows the exponential decay of the option notional ($dX_t = -q X_t dt$). All option premia, rebates, and market-operation yields are paid in the num\'eraire asset.

\begin{remark}
For expositional simplicity, herein we will consider these AmPO tokens as \emph{negatively rebasing} tokens so that the notional and total open interest would be equivalent. However, in practice, we recommend considering a wrapped version of these tokens so that an internal accounting handles the rebasing rather than modifying users' token holdings. That is, a global pricing index $e^{-qt}$ can be used to reduce the value of AmPO tokens over time. 
These mechanisms are the inverse of those used in liquid staking protocols~\cite{gogol2024sok}, i.e., the positive rebasing found in, e.g., Lido or reward-bearing tokens found in, e.g., RocketPool.
\end{remark}

\section{Decentralized Market Design}\label{sec:design}
While Section~\ref{sec:options} introduced the AmPO contract as a perpetual derivative that satisfies the required axiomatic constraints of the blockchain, the viability of this contract in DeFi relies upon a decentralized primary market. Within this section, we propose a peer-to-pool market design for the underwriting and issuance of AmPOs. Herein, the underlying mathematics of the market are proposed (Section~\ref{sec:ampo-design}), the market operations are enumerated (Section~\ref{sec:ampo-market}), and the properties of this primary market are presented (Section~\ref{sec:ampo-properties}). 

\subsection{The Pricing Function}\label{sec:ampo-design}
Consider now a fixed AmPO series (type, strike $K$, and amortization rate $q$) with collateral held in the asset to be delivered from physical settlement as discussed in Section~\ref{sec:options}. To guarantee solvency of the pool, the open interest for this AmPO series must not exceed the total collateral held. In this way, we can define the option premium via the utilization $U := X/C \in [0,1]$ rather than as a bivariate function.\footnote{Recall, for put options, that the collateral posted in the num\'eraire asset is converted into notional-equivalent units using the strike $K$.} This construction follows the legacy of lending protocols which construct interest rate curves as explicit functions of the utilization~\cite{leshner2019compound,gudgeon2020defi}.
\begin{definition}\label{defn:premium}
A mapping $P: [0,1] \to \R_+ \cup \{\infty\}$ is called the \textbf{\emph{premium function}} if it is non-decreasing with $P(U) > 0$ for any $U > 0$ and
\[\lim_{U \nearrow 1} P(U) = \begin{cases} \infty &\text{for call options,} \\ K &\text{for put options with strike $K$}. \end{cases}\]
\end{definition}
\begin{remark}\label{rem:premium-limit}
The properties of the premium function encode the minimal set of assumptions needed for a meaningful option pricing. 
Monotonicity and $P(U) > 0$ for $U > 0$ enforce that the premium charged should be non-trivial and never decrease. That is, $P$ provides a meaningful supply curve for AmPOs.
The limiting behavior imposed, which differs for call and put options, tracks the supremum of the payoff for these two option types. By a simple no-arbitrage argument, since a put option has payoff $(K-S)^+ \leq K$, a rational actor would never pay more than the strike for this option; if $P(U) > K$ for some $U < 1$ then that would guarantee some collateral is idle and reduces the capital efficiency of the system (compare to, e.g., the idea of ``no wasted liquidity'' from \cite{bichuch2022axioms}).
Furthermore, because AmPOs are perpetual options, the premium function is time homogeneous and, therefore also, time independent thus simplifying the construction. 
\end{remark}

\begin{example}\label{ex:premium}
To make the premium function more concrete, we want to consider specific choices of such a function. 
\begin{itemize}
\item \textbf{Call option:} Consider the premium function $P_\text{call}(U) = 2U/(1-U)^3$ for $U \in [0,1)$ with $P_\text{call}(1) = \infty$. 
\item \textbf{Put option} with strike $K$: The linear premium function $P_\text{put}(U) = KU$ for $U \in [0,1]$ is a simple construction. 
\end{itemize}
\end{example}

While the premium function provides the cost of the ``next'' option (with marginal notional), the realized cost of an actual position is based on the integral of the premium function. We encode this in the net premium function below. Notably, as provided in Proposition~\ref{prop:net-premium}, there is an equivalence between the premium and net premium functions so that either can be treated as the primary driver of the options market. 

\begin{definition}\label{defn:net-premium}
Given a premium function $P$, the running integral mapping $\phi_P(U) := \int_0^U P(u)du$ is the \textbf{\emph{net premium function}}.
\end{definition}
\begin{proposition}\label{prop:net-premium}
Given a premium function $P$, the net premium function $\phi_P$ is differentiable a.e., convex, and strictly increasing with $\phi_P(0) = 0$.
Conversely, given a differentiable a.e., convex, and strictly increasing mapping $\phi: [0,1] \to \R_+ \cup \{\infty\}$ with $\phi(0) = 0$ and $$\lim_{U \nearrow 1} \phi'(U) = \begin{cases} \infty &\text{for call options}, \\ K &\text{for put options with strike $K$},\end{cases}$$ there exists a premium function $P_\phi(U) := \phi'(U)$.
Furthermore, under suitable continuity assumptions, this relation is one-to-one, i.e., $\phi_{P_\phi} \equiv \phi$ (if $\phi$ is differentiable) and $P_{\phi_P} \equiv P$ (if $P$ is continuous).
\end{proposition}
\begin{proof}
First, given a premium function $P$, the properties of the net premium function $\phi_P(U) := \int_0^U P(u)du,~U\in[0,1]$ follow directly by construction. 
Consider, now, a differentiable, convex, and strictly increasing mapping $\phi: [0,1] \to \R_+ \cup \{\infty\}$ with boundary conditions as in the statement of this proposition. Define $P_\phi(U) := \phi'(U),~U\in[0,1),$ and $P_\phi(1) :=\lim\limits_{U\nearrow 1} P_\phi(U)$. By monotonicity and convexity, it must follow that $P_\phi$ is non-decreasing and non-negative. In fact, by $\phi$ strictly increasing and convex, $P_\phi(U) = \phi'(U) > 0$ for any $U > 0$. 
The limiting behavior $\lim_{U \nearrow 1} P_\phi(U)$ follows directly by construction.
Finally, the one-to-one mapping between the premium function $P$ and the net premium function follows from the fundamental theorem of calculus.
\end{proof}

\begin{example}[Example~\ref{ex:premium} continued]\label{ex:net-premium}
Let the premium function $P$ be as Example~\ref{ex:premium}. Then the associated net premium functions are: 
\begin{itemize}
\item \textbf{Call option:} $\phi_\text{call}(U) = [U/(1-U)]^2$ is quadratic in the odds ratio. 
\item \textbf{Put option} with strike $K$: $\phi_\text{put}(U) = KU^2/2$ is quadratic.
\end{itemize}
\end{example}

The net premium function, either explicitly defined as the integral of the premium function over the utilization or defined as in Proposition~\ref{prop:net-premium}, can be used to define the actualized cost (or rebate) charged to market participants based on the actual open interest ($X$) and collateral ($C$). This is encoded in the total premium function, which acts in much the same way as the cost function in prediction markets~\cite{hanson2007logarithmic,chen2012utility}.
\begin{definition}\label{defn:total-premium}
Given a net premium function $\phi$, the \textbf{\emph{total premium function}} is the mapping 
\[\Phi(X,C) := \begin{cases} C \phi(X/C) &\text{if } X \leq C, \\ \infty &\text{else}.\end{cases}\] 
By continuity, we set $\Phi(0,0) = 0$.
\end{definition}
\begin{proposition}\label{prop:total-premium}
Let $\Phi$ be a total premium function. Then $\Phi$ is strictly increasing in its first component, non-increasing in its second component, convex, positive homogeneous, and differentiable a.e.
\end{proposition}
\begin{proof}
\begin{itemize}
\item \textbf{Strictly increasing in the first component:} This follows directly from strict monotonicity of the net premium function $\phi$.
\item \textbf{Non-increasing in the second component:} Fix $X \geq 0$. If $X = 0$ then $\Phi(0,C) = 0$ for any $C \geq 0$. Take $X > 0$ and assume $C \geq X$ (as monotonicity is trivial in $C < X$ and when crossing the $X = C$ boundary from the left). By convexity of the net premium function $\phi$, and using the fact that $\phi(0)=0$, the result follows by taking the derivative of $\Phi$ w.r.t.\ $C$:
\[\frac{\partial}{\partial C} \Phi(X,C) = \phi(X/C) - \frac{X}{C}\phi'(X/C) \leq 0.\]
\item \textbf{Convex:} This follows directly from $\Phi$ being the perspective of the net premium function $\phi$ (see, e.g.,~\cite[Chapter 3.2.6]{boyd2004convex}).
\item \textbf{Positive homogeneity:} Fix $X \geq 0$, $C \geq 0$, and $t > 0$. If $C = 0$ then $\Phi(tX,tC) = \Phi(tX,0) = t\Phi(X,0)$ as these are either all $0$ (if $X = 0$) or $\infty$ (if $X > 0$). Similarly, if $X > C > 0$ then $\Phi(tX,tC) = \infty = t\Phi(X,C)$. Finally, take $C > 0$ and $X \leq C$. Immediately, $\Phi(tX,tC) = tC\phi([tX]/[tC]) = tC\phi(X/C) = t\Phi(X,C)$.
\item \textbf{Differentiable:} This follows trivially.
\end{itemize}
\end{proof}

\begin{remark}\label{rem:premium-limit-2}
While the premium for call options limits to $\infty$ as the market reaches full utilization ($X \nearrow C$), the same need not hold for the total premium function. In particular, it is possible that $\Phi(t,t) < \infty$ for any $t \geq 0$. Under such a case, it is possible for an investor to push the marginal cost of an option to be infinite for a finite cost. To guarantee the collateral can never be fully utilized (guaranteeing a stronger version of solvency), we can impose the additional constraint that $\lim_{U \nearrow 1} \phi(U) = \infty$ as in Example~\ref{ex:net-premium}.
\end{remark}

\begin{example}[Example~\ref{ex:premium} continued]\label{ex:total-premium}
Let the net premium function $\phi$ be as Example~\ref{ex:net-premium}. Then the associated total premium functions are:
\begin{itemize}
\item \textbf{Call option:} $\Phi_\text{call}(X,C) = C[X/(C-X)]^2$ follows the quadratic of the ratio of open interest to idle collateral on $X < C$. 
\item \textbf{Put option} with strike $K$: $\Phi_\text{put}(X,C) = [KX^2]/[2C]$ is quadratic on $X \leq C$.
\end{itemize}
\end{example}

\subsection{Market Operations}\label{sec:ampo-market}

The total premium function $\Phi: \R^2_+ \to \R_+ \cup \{\infty\}$ can be used to define the cost (or rebate) of any market operation. Briefly, given current open interest notional $X \geq 0$ and collateral $C \geq 0$ and an operation that alters these market parameters $x \in [-X,\infty)$ and $c \in [-C,\infty)$ then, analogously to the cost function in prediction markets~\cite{hanson2007logarithmic,chen2012utility}, the total cost $\Delta\Phi$ of the operation is:
\begin{equation}\label{eq:cost}
\Delta\Phi(X,x,C,c) := \Phi(X+x,C+c) - \Phi(X,C).
\end{equation}
In the case that $\Delta\Phi(X,x,C,c) < 0$ then the market operation provides a rebate in the amount of $|\Delta\Phi(X,x,C,c)|$ to the trader instead of representing a cost to be paid. As will be formally proven in Section~\ref{sec:ampo-properties} below, the total premium held by the market at any time is $\Phi(X,C)$; in this way, all costs are paid to the market contract while rebates come explicitly from this contract. 
This cost function immediately maps to the four primary market operations for options, i.e., to buy/sell to open/close. Throughout this discussion we fix the market parameters $X \geq 0$ and $C \geq 0$.
\begin{itemize}
\item \textbf{Buy to Open:} When an investor wants to purchase $x \in [0,C-X]$ notional-units of options, he or she must pay $\Delta\Phi(X,x,C,0) \geq 0$ for these options.
\item \textbf{Sell to Close:} When an investor wants to sell back $x \in [0,X]$ notional-units of options, he or she receives $-\Delta\Phi(X,-x,C,0) \geq 0$ in num\'eraire tokens for burning these options and reducing pool utilization.
\item \textbf{Sell to Open:} When an investor wants to provide liquidity $c \geq 0$ to underwrite new options, he or she receives $-\Delta\Phi(X,0,C,c) \geq 0$ in num\'eraire tokens immediately for this underwriting. Notably, if $X = 0$ (i.e., for the initial market underwriter) then this rebate is 0; we will discuss below how such a liquidity provider can be compensated for his or her risk.
\item \textbf{Buy to Close:} When a liquidity provider wants to close his or her position $c \in [0,C-X]$, he or she must pay $\Delta\Phi(X,0,C,-c) \geq 0$ to limit the associated risk. We note that a liquidity provider is restricted to closing at most $C-X \geq 0$ options with finite cost even if they have provided a greater amount of collateral than that limit.
\end{itemize}

While all of the aforementioned four market operations have costs and rebates for individual investors, there are two more types of operations that distribute yield to all liquidity providers.
\begin{itemize}
\item \textbf{Exercise Yield:} When an option holder exercises $x \in [0,X]$ notional-units of options, he or she pays the requisite amount for the physical delivery of $x$ notional-units of collateral; this payment should be distributed proportionally to all underwriters. However, in addition to this explicit payment, the exercise event reduces the market's utilization as the open interest is decreased by the same amount as the collateral (i.e., $[X-x]/[C-x] \leq X/C$). In this way, an exercise rebate in the amount of $-\Delta\Phi(X,-x,C,-x) \geq 0$ is also proportionally distributed to the underwriters in compensation for the risk they have taken.
\item \textbf{Amortization Yield:} In addition to rebates earned from \emph{sell to open} and \emph{exercise} events, underwriters of AmPOs earn a continuous, passive yield from the amortization of all outstanding options. As described in Section~\ref{sec:options}, the cost of holding an AmPO is paid implicitly by the option holders; at every moment, a small fraction of their option's notional is amortized away. This decrease in open interest acts as a cash transfer to the underwriters. This process creates a continuous stream of value flowing from the collective of option holders to the liquidity providers who are underwriting the market's risk. This amortization yield forms a core and predictable component on an underwriter's total return. Notably, the higher the premium (and, therefore, also the risk for underwriters), the higher this continuous liquidity provider yield. Here, consider the open interest at time $t_0$ to be $X \geq 0$, then the proportionally distributed amortization yield, assessed at time $t \geq t_0$, is $-\Delta\Phi(X,-[1-e^{-q(t-t_0)}]X,C,0) \geq 0$.
\end{itemize}

\begin{remark}
As discussed with \emph{buy to close}, an option underwriter may be limited in being able to close their position if there is a sufficiently high open interest. However, due to the structure of the amortization yield, when an underwriter may be restricted from closing his or her position is exactly the same time when this passive yield is highest. That is, when the collateral $C$ is fixed, the yield is highest when the open interest $X \leq C$ is very close to $C$. 
\end{remark}

\begin{remark}
The yield for liquidity providers from \emph{sell to open}, \emph{exercise}, and \emph{amortization} is distinct from that earned by liquidity providers on a spot-market AMM. At these traditional AMMs, liquidity providers \emph{only} earn yield from transaction costs~\cite{angeris2020improved} while being subject to impermanent loss~\cite{xu2023sok} or loss-versus-rebalancing~\cite{milionis2022automated} on the underlying position. Herein we only consider these natural economic yields, though transaction costs can be included in the market mechanisms to further incentivize underwriters.
\end{remark}

\begin{remark}\label{rem:async}
To strictly satisfy asynchronicity as necessitated by the blockchain, the amortization yield is only computed upon some user interaction (market operations, exercise, or an explicit yield redemption). Crucially, because amortization follows the exponential function, the accumulated yield over any discrete time interval is path independent, i.e., the yield computed in a single atomic update after $\Delta t > 0$ time has passed is identical to the sum of yields computed over any partition of $\Delta t$ (assuming no other interaction has occurred). In this way, the frequency of interactions does not alter the resulting economic outcomes.  Notably, this property introduces a form of path independence in the time-dimension that is distinct from the market-state dimension that is commonly considered for such properties (see, e.g.,~\cite{schlegel2022axioms,bichuch2022axioms} for such a discussion with AMMs).
\end{remark}

\subsection{Properties}\label{sec:ampo-properties}
In Section~\ref{sec:options}, we discussed the AmPO contract and, with physical settlement, demonstrated its suitability for all five stated design constraints. We now wish to consider the mathematical properties of the market designed in Sections~\ref{sec:ampo-design} and~\ref{sec:ampo-market}. In particular, we will prove that the proposed market design guarantees solvency of the system as well as path independence for market operations. 
For simplicity of the discussion, throughout this section we consider Assumption~\ref{ass:blocks} to enforce conditions of the blockchain on the evolution of the market processes $(X_t)_{t \geq 0},(C_t)_{t \geq 0}$.

\begin{assumption}\label{ass:blocks}
Market events occur at discrete (potentially random) block times $(\tau_i)_{i \in \mathbb{N}}$ with $\tau_0 := 0$ and $\tau_i < \tau_{i+1}$ and market events are executed sequentially within each block. 
Furthermore, we impose the lazy amortization system described in Remark~\ref{rem:async} so that the market processes $(X_t)_{t \geq 0}$ and $(C_t)_{t \geq 0}$ are c\`adl\`ag step processes with jumps occurring the end-of-block open interest and collateral, respectively, at block times. 
In addition, we will assume that the market is initialized from $X_0 := 0$ and $C_0 := 0$.
\end{assumption}

Notably, though the market processes $(X_t)_{t \geq 0}$ and $(C_t)_{t \geq 0}$ are described as updating to the end-of-block values in Assumption~\ref{ass:blocks}, within a block multiple transactions can occur. In the following proposition, we prove an initial result towards path independence so that the market state after a block only depends on its end-of-block state.
\begin{proposition}\label{prop:intra-block}
Consider the blockchain system defined in Assumption~\ref{ass:blocks}. Consider block $i \in \mathbb{N}$ at time $\tau_i$. The market reserves $R_{\tau_i}$ at the end of the sequence of (countable) transactions within block $i$ is equal to $R_{\tau_{i-1}} + \Phi(X_{\tau_i},C_{\tau_i}) - \Phi(X_{\tau_{i-1}},C_{\tau_{i-1}})$.
\end{proposition}
\begin{proof}
Fix block $i$ and assume there are $N \in \mathbb{N} \cup \{\infty\}$ (countable) operations within that block. Denote these \emph{valid} operations $(x_n,c_n)_{n = 1}^N$ and the resulting market parameters are $X_{\tau_i}^n = X_{\tau_i}^{n-1} + x_n \geq 0$ and $C_{\tau_i}^n = C_{\tau_i}^{n-1} + c_n \geq 0$ with $X_{\tau_i}^0 := X_{\tau_{i-1}}$ and $C_{\tau_i}^0 := C_{\tau_{i-1}}$.\footnote{Validity of the market operations is defined so that $X_{\tau_i}^n \geq 0$, $C_{\tau_i}^n \geq 0$, and $\Phi(X_{\tau_i}^n,C_{\tau_i}^n) < \infty$ for every operation $n = 1, ..., N$.} Furthermore, assume convergence of these operations to the end-of-block parameters $X_{\tau_i} = \lim_{n \nearrow N} X_{\tau_i}^n$ and $C_{\tau_i} = \lim_{n \nearrow N} C_{\tau_i}^n$. (In the case that $N < \infty$ then these limiting cases can be reduced directly to $X_{\tau_i}^N,C_{\tau_i}^N$ respectively.)
With this construction, the desired result follows from the telescoping sum based on the definition of the net inflows $\Delta\Phi$ as provided in~\eqref{eq:cost}:
\begin{align*}
R_{\tau_i} &= R_{\tau_{i-1}} + \sum_{n = 1}^N \Delta\Phi(X_{\tau_i}^{n-1},x_n,C_{\tau_i}^{n-1},c_n) \\
&= R_{\tau_{i-1}} + \sum_{n = 1}^N \left[\Phi(X_{\tau_i}^{n-1}+x_n,C_{\tau_i}^{n-1}+c_n) - \Phi(X_{\tau_i}^{n-1},C_{\tau_i}^{n-1})\right] \\
&= R_{\tau_{i-1}} + \sum_{n = 1}^N \left[\Phi(X_{\tau_i}^n,C_{\tau_i}^n) - \Phi(X_{\tau_i}^{n-1},C_{\tau_i}^{n-1})\right] \\
&= R_{\tau_{i-1}} + \lim_{n \nearrow N} \Phi(X_{\tau_i}^n,C_{\tau_i}^n) - \Phi(X_{\tau_i}^0,C_{\tau_i}^0)\\
&= R_{\tau_{i-1}} + \Phi(X_{\tau_i},C_{\tau_i}) - \Phi(X_{\tau_{i-1}},C_{\tau_{i-1}}).
\end{align*}
\end{proof}

Following Proposition~\ref{prop:intra-block}, the market system evolves solely based on its end-of-block values. We now expand that intra-block result to consider the inter-block reserves.
\begin{corollary}\label{cor:inter-block}
Consider the market processes $(X_t)_{t \geq 0}$ and $(C_t)_{t \geq 0}$ in open interest and collateral, respectively, satisfying Assumption~\ref{ass:blocks} and with $R_0 := 0$ in initial reserves used to pay out any obligations. 
Then path independence holds, i.e., $R_{t_1} - R_{t_0} = \Phi(X_{t_1},C_{t_1}) - \Phi(X_{t_0},C_{t_0})$ for any times $0 \leq t_0 < t_1$ without dependence on the realized path from $(X_{t_0},C_{t_0})$ to $(X_{t_1},C_{t_1})$.
In particular, at any time $t \geq 0$, the capital reserves held by the market contract (i.e., the cumulative net inflow of funds to the contract) satisfies $R_t = \Phi(X_t,C_t)$.
\end{corollary}
\begin{proof}
Fix times $t_1 > t_0 \geq 0$. Take $N_1 = \max\{i \in \mathbb{N} \; | \; \tau_i \leq t_1\}$ and $N_0 = \max\{i \in \mathbb{N} \; | \; \tau_i \leq t_0\}$. Recall that $\tau_0 := 0$ for simplicity of notation. If $N_1 = 0$ then the result follows directly from initialization of the processes. Assume $N_1 \geq 1$. Since the system follows a step process, it trivially follows that $R_{t_k} = R_{\tau_{N_k}}$, $X_{t_k} = X_{\tau_{N_k}}$, $C_{t_k} = C_{\tau_{N_k}}$ for $k \in \{0,1\}$. Following the recursive construction in Proposition~\ref{prop:intra-block}, the result trivially follows from the telescoping sum:
\begin{align*}
R_{t_1} &= R_{\tau_{N_1-1}} + \Phi(X_{\tau_{N_1}},C_{\tau_{N_1}}) - \Phi(X_{\tau_{N_1-1}},C_{\tau_{N_1-1}}) = \cdots \\
&= R_{\tau_{N_0}} + \sum_{i = N_0+1}^{N_1} \left[\Phi(X_{\tau_i},C_{\tau_i}) - \Phi(X_{\tau_{i-1}},C_{\tau_{i-1}})\right] \\
&= R_{\tau_{N_0}} + \left[\Phi(X_{\tau_{N_1}},C_{\tau_{N_1}}) - \Phi(X_{\tau_{N_0}},C_{\tau_{N_0}})\right] = R_{t_0} + \left[\Phi(X_{t_1},C_{t_1}) - \Phi(X_{t_0},C_{t_0})\right].
\end{align*}
Finally, $R_t = \Phi(X_t,C_t)$ can be seen by considering $t_1 = t$ and $t_0 = 0$, noting that $R_0 = 0 = \Phi(0,0)$.
\end{proof}

\begin{remark}
As highlighted in, e.g.,~\cite{bichuch2022axioms}, path independence implies strong no-arbitrage conditions. For instance, under path independence, there does not exist a round-trip trade in which an investor can extract positive value from the system (e.g., buying then selling an option leads to zero net-transfers). If fees are introduced, such a property results in strict losses for any attacker trying to make such a round-trip transaction.
\end{remark}

\begin{corollary}\label{cor:solvent+path}
Consider the setting of Assumption~\ref{ass:blocks}.
The market is solvent at all times $t \geq 0$, i.e., $R_t \geq 0$, and the reserves depend only on the current market state $(X_t,C_t)$.
\end{corollary}
\begin{proof}
These results follow from a direct application of Corollary~\ref{cor:inter-block} as $R_t = \Phi(X_t,C_t) \geq 0$ by construction.
\end{proof}

\begin{remark}
In Corollary~\ref{cor:solvent+path} we express solvency solely in terms of being able to pay out any market-based obligations from the market reserves. Full contractual solvency, which additionally requires the ability to make the necessary payments at option exercise, is guaranteed by the physical settlement and Exercise Yield procedures along with the construction so that the open interest is always bounded by the collateral, i.e., $X \leq C$.
\end{remark}

\section{Upgrading the DeFi Stack}\label{sec:apps}
\subsection{Protective Puts and Endogenous Haircuts in Lending Markets}\label{sec:apps-lending}

\begin{figure}
\centering
\begin{tikzpicture}
\begin{axis}[
    xlabel = {Time $t$ (days)},
    xmin = 0, xmax = 10,
    ymin = 0.995, ymax = 1.025,
    domain = 0:10,
    samples = 200,
]
\addplot [red,ultra thick,name path=debt] {exp(0.000137*x)}; 
\addplot [black,ultra thick,name path=liq] {1.01*exp(0.000137*x)}; 
\addplot [blue,ultra thick,name path=collat] {1.02*exp(-0.003*x)}; 
\addplot [black,dotted,ultra thick] coordinates {(3.14067,0.985) (3.14067,1.03)};
\addplot [black,dotted,ultra thick] coordinates {(6.3126,0.985) (6.3126,1.03)};
\path [name path=top103] (0,1.03) -- (10,1.03);
\path [name path=bot0985] (0,0.985) -- (10,0.985);
\addplot [green!20!white] fill between[of=bot0985 and top103, soft clip={domain=0:3.14067}];
\addplot [yellow!20!white] fill between[of=bot0985 and top103, soft clip={domain=3.14067:6.3126}];
\addplot [red!20!white] fill between[of=bot0985 and top103, soft clip={domain=6.3126:10}];
\end{axis}
\end{tikzpicture}
\caption{Visualization of liquidation events due to amortization. The blue line is 30bps/day amortizing collateral, the black line is the liquidation threshold for the loan growing at 5\% annualized, and the red line is the actual obliged debt growing at 5\% annualized. The green region indicates no liquidations, the yellow region is the region with safe liquidations, and the red region (if reached) is bad debt.}
\label{fig:liquidation-event}
\end{figure}
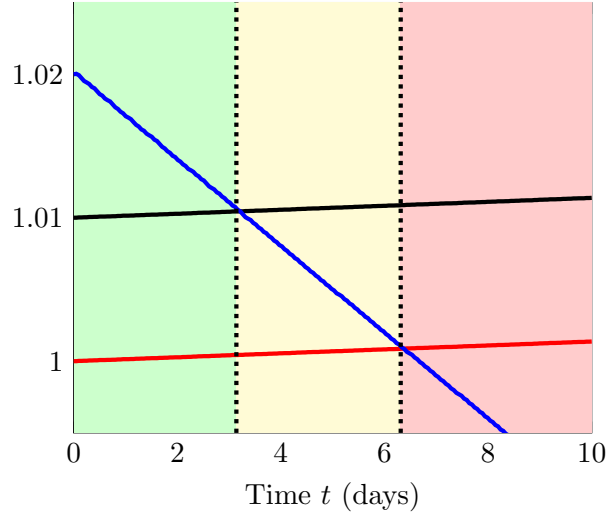

Lending markets, such as Aave, Compound, and Morpho, rival liquid staking protocols for the largest sector of DeFi by total value locked in late 2025.\footnote{\url{https://defillama.com/}} Largely, these markets all function in the same way. As summarized in, e.g.,~\cite{gudgeon2020defi}, liquidity providers lend out cryptocurrencies (e.g., USDC) to borrowers who overcollateralize their positions with other tokens (e.g., ETH). If the loan-to-value ratio exceeds a governance set threshold, liquidators can claim the collateral in exchange for closing out the loan. So long as the liquidators can act fast enough and the market remains sufficiently liquid, the liquidators are able to make a profit while also keeping the lenders whole. However, because the loan-to-value threshold is exogenous and slow to adapt, bad debt can accumulate when no liquidator can profitably act.

This reliance on fixed, exogenous haircuts for different cryptocurrencies (i.e., the loan-to-value thresholds) leads to systemic fragility. If the haircut is too large, then the lending market becomes excessively capital inefficient; if the haircut is too small, then the lenders are subject to the risk of bad debt. Furthermore, because the value of the collateral requires external oracles, the loan-to-value ratio can be temporarily manipulated leading to novel attacks (see, e.g.,~\cite{qin2021attacking} which describes one such event from February 18, 2020). Herein we investigate how AmPOs permit an alternative collateral system in which borrowers, effectively, pre-fund their own liquidation by pairing their collateral with put options.

For simplicity of exposition, consider a user who is borrowing USDC from Aave using WETH as collateral. As of this writing, this user can borrow 80.5\% of the value of the WETH and will be liquidated if the loan-to-value ratio exceeds 83\% at Aave v3.\footnote{\url{https://app.aave.com/markets/}}
Consider, alternatively, utilizing an AmPO protective put position with strike $K$ USDC and amortization rate $q > 0$ as collateral to borrow up to $K$ USDC per unit of collateral without introducing WETH-based risks. 
That is, the user supplies the original $E > 0$ tokens of WETH coupled with $\alpha E$ put options for $\alpha > 1$. Over time, due to the amortization, the protection decays following the negative exponential $t \mapsto \alpha e^{-qt}E$. 
In this way, the collateral can be valued based solely on the option notional $\alpha e^{-qt}EK$ USDC as the exercise value of this protective put is always $\min\{\alpha e^{-qt},1\} EK$. As this collateral value is decreasing predictably over time, and the loan value is increasing due to the charged interest, the standard liquidation procedure still occurs if the loan-to-value $(\alpha e^{-qt}E)^{-1}$ gets sufficiently close to 1.
A stylized visualization of the possible state of the collateral is provided in Figure~\ref{fig:liquidation-event}. 

The use of protective put collateral resolves three fundamental vulnerabilities inherent in current liquidation mechanisms. First, it eliminates the price impacts triggered by liquidation events; rather than selling the seized collateral in the spot market, the physically-settled put options guarantee a swap at the strike price $K$ without touching the underlying market. In this way, the protective put structure pre-funds the liquidation, protecting lenders from the market risk of a failed liquidation. Second, this collateral achieves oracle independence because the value of the protective put can be determined solely by the exponential decay of the option's notional without reference to the spot price. In doing so, borrowers and lenders are insulated from oracle manipulations and flash-crash liquidations. Third, the loan-to-value ratio threshold becomes endogenously priced by market supply and demand via the option premium instead of exogenously determined by governance votes (e.g., the current 19.5\% haircut for WETH). Taken together, protective put collateral constructs a system in which lenders are \emph{not} subject to market or oracle risks related to the collateral. 

\begin{remark}
While AmPO underwriting is naturally capital inefficient, requiring full collateralization of every option underwritten, the composability with lending protocols makes token issuers natural underwriters for AmPO markets. That is, a DAO can use a fraction of its (stablecoin) treasury in order to underwrite AmPOs on, e.g., its governance token. Doing so provides four distinct structural benefits for the DAO:
    \begin{enumerate}
    \item it creates immediate utility for the underlying token by enabling its use as viable collateral (in contrast to the near 100\% haircuts traditionally imposed on long-tail assets);
    \item the ability to borrow against the underlying token reduces selling pressure, providing an alternative means for holders to access liquidity;
    \item the chosen strike price acts as an explicit reserve price for the token, replacing the need for ad-hoc governance interventions during market stress events; and
    \item the underwriting position generates yield for the DAO's treasury via the amortization rate.
    \end{enumerate}
In this way, the DAO internalizes the cost of capital inefficiency as a necessary operational expense to defend and empower its own ecosystem.
\end{remark}

\subsection{De-Peg Insurance: Stablecoins and Liquid Staking}\label{sec:apps-stable}

Stablecoins are, arguably, the most successful application in DeFi at present with over \$300 \emph{billion} in total stablecoins market cap as of late 2025.\footnote{\url{https://defillama.com/stablecoins}} Despite this success, there have been notable de-pegging events in recent years; the collapse of Terra UST led to a broad selloff that precipitated the years-long ``crypto winter''~\cite{oecd2022lessons} while USDC broke its peg in March 2023 due to the failure of Silicon Valley Bank~\cite{diop2024collapse}. 
In DeFi, holders of stablecoins often lack financial protection against these events. In comparison, in traditional finance, institutions buy put options to hedge against currency risks. While DeFi does offer insurance-like products at, e.g., Nexus Mutual, these products utilize a discretionary claims process~\cite{nexus} or external oracles subject to manipulation~\cite{eskandari2021sok}. As far as the authors are aware, no fungible and autarkic put option protection is available in DeFi as of this writing.

\begin{remark}
While distinct, liquid staking tokens serve a similar function to stablecoins in that they are meant to keep a peg to the underlying asset. For instance, stETH issued by Lido is meant to maintain a 1-to-1 peg to ETH while rebasing to provide staking rewards for investors. However, stETH will regularly lose its peg due to, e.g., a withdrawal queue from Ethereum staking~\cite{gogol2024sok}. Much like how the fiat-backed USDC can break its peg, analogous risks can exist with Lido due to, e.g., a hack that drains the smart contract of funds. Thus, though herein we will generally refer to stablecoins due to their market capitalization, the same de-peg insurance can be applied to liquid staking tokens.
\end{remark}

Though not an explicit put option, in December 2020, MakerDAO (now Sky) introduced the Peg Stability Module (PSM)~\cite{kozhan2021decentralized}. This mechanism allowed investors to freely convert DAI (and USDS) to USDC at parity; implicitly this acts as a put option \emph{without} a cost-of-carry as MakerDAO does not charge users based on the convexity of this protection. In this way, the cost of offering this option is being borne directly by the DAO and its equity holders.

Consider, now, a version of the PSM in which the cost-of-carry is imposed, i.e., the end-users pay for their stability protection. For simplicity of exposition, we will consider this setting with MakerDAO as an example. That is, consider a situation in which MakerDAO acts as a liquidity provider, depositing its reserves of USDC as collateral in an at-the-peg put option market with DAI as the underlying asset. Now, if a user wants to convert their DAI to USDC, they need to hold the put option and exercise it via the physical delivery mechanism. In this way, MakerDAO is able to be compensated for the service they are providing; in particular, this compensation is paid via the amortization yield earned by underwriters in the AmPO mechanism.
In fact, by making the PSM an explicit option, this protection can be underwritten by third-party investors as well. This decentralizes the insurance currently being offered by MakerDAO, potentially leading to greater liquidity and capital efficiency. 

Furthermore, this protection, being an explicit option, additionally provides useful forward information about the implicit risk to DAI by referencing the option premium. This premium provides an on-chain ``fear index'' for the stablecoin which is independent of any external oracles. This is in contrast to current DeFi practice where, e.g., the price of risk for UST was unknown until its peg was already irrevocably broken.

\subsection{Decentralized Clearinghouse}\label{sec:apps-clearinghouse}

Following the financial crisis of 2007–2009, regulators in Traditional Finance (TradFi)  moved toward clearing derivatives through Central Counterparty Clearing Houses (CCPs) to minimize credit risk and manage default losses~\cite{duffie2011does}. By acting as the counterparty in every transaction, CCPs facilitate novation, netting, and the mutualization of tail losses across their participating institutions. In contrast, DeFi operates on a fragmented architecture where every protocol functions as its own bespoke clearinghouse. Utilizing a peer-to-pool model, each protocol implements isolated collateral rules, liquidation logic, and risk buffers~\cite{gudgeon2020defi}.

Consequently, users are exposed to the specific default risk of every protocol they interact with. Because DeFi lacks a native mechanism for cross-protocol mutualization, these risks are siloed within individual smart contracts. For instance, during a systemic market shock, every independent lending protocol and stablecoin issuer can only rely on its network of liquidators and insurance fund to protect liquidity providers. This means that these systemic risks are duplicated across the ecosystem while, simultaneously, being inconsistently priced~\cite{angeris2020improved}. While the TradFi solution of a CCP offers default management and loss mutualization, it inherently contradicts the core DeFi principles of distribution and decentralization~\cite{pirrong2011economics}.

AmPOs offer a novel, decentralized solution to this dilemma~\cite{feinstein2025amortizing}. By offering a standardized derivative contract that covers a diverse range of protocol risks, AmPOs act as a universal risk transfer mechanism. In particular, if multiple protocols utilize AmPO-based structures (such as lending markets and stablecoin issuers), then liquidation risk, de-peg risk, and solvency backstops are all translated into the same option primitive. 
In this way, the convex tail risk is absorbed by underwriters in a global, shared pool of liquidity. That is, by facilitating risk transfer to these liquidity providers, AmPOs allow risk to be aggregated and mutualized across distinct use cases and protocols. Furthermore, as tradable derivatives, the convex tail risk becomes explicitly priced, fungible, and transferable across the blockchain. Ultimately, this approach achieves the mutualization benefits of a clearinghouse without introducing centralized institutions.

\section{Conclusion}\label{sec:conclusion}
In this work, we presented a DeFi-native option design tailored specifically to the operational and adversarial constraints of blockchain environments. We noted five fundamental axioms that any option must satisfy to be compatible with DeFi, all of which are satisfied by AmPOs. This is in contrast to most prior DeFi derivative designs which attempt to replicate centralized exchange constructions, failing these core axioms by requiring, e.g., high-frequency oracles. With a focus on AmPOs, we proposed a decentralized, peer-to-pool primary market framework for the underwriting and trading of these instruments.

Following the composability axiom, we considered example applications of AmPOs within the broader DeFi ecosystem. In considering these options as a foundational risk primitive, we presented the implications of embedding a hedging instrument into lending protocols (i.e., protective put collateral) and stablecoins (i.e., de-peg insurance). In applying the same primitive throughout DeFi, we found that AmPOs provide a universal layer for mutualizing tail risk across protocols without implementing a centralizing institution. Notably, though we have proposed the idea of the decentralized clearinghouse herein, we believe that such a construction should be the subject of further study.

Within the work, the AmPO market mechanism was presented without quantifying the underwriter (liquidity provider) profitability. Future work can consider, e.g., the loss-versus-rebalancing (LVR)~\cite{milionis2022automated} of the underwriters in this primary market. In comparison to the LVR on spot markets, the LVR for options would compare the liquidity provider profits and losses to that of a replicating strategy in the underlying spot market. 

Finally, while AmPOs provide a perpetual option design that adheres to the presented axioms, they are not the only viable design for on-chain derivatives. Future work can investigate alternative autarkic option designs. In particular, as noted within~\cite{amini2025decentralized}, prediction markets can be used to create derivative contracts; if paired with an underlying AMM-based spot market, the joint AMM-prediction market system could be constructed to be autarkic and, therefore, also satisfy all proposed axioms even as a European option.

\bibliographystyle{apalike}
\bibliography{bibtex}

\begin{thebibliography}{}

\bibitem[Adams et~al., 2021]{adams2021uniswap}
Adams, H., Zinsmeister, N., Salem, M., Keefer, R., and Robinson, D. (2021).
\newblock Uniswap v3 core.
\newblock {\em Uniswap, Tech. Rep.}

\bibitem[Amini et~al., 2025]{amini2025decentralized}
Amini, H., Bichuch, M., and Feinstein, Z. (2025).
\newblock Decentralized prediction markets and sports books.
\newblock {\em Mathematical Finance}.

\bibitem[Angeris and Chitra, 2020]{angeris2020improved}
Angeris, G. and Chitra, T. (2020).
\newblock Improved price oracles: Constant function market makers.
\newblock In {\em Proceedings of the 2nd ACM Conference on Advances in
  Financial Technologies}, pages 80--91.

\bibitem[Angeris et~al., 2023a]{angeris2023geometry}
Angeris, G., Chitra, T., Diamandis, T., Evans, A., and Kulkarni, K. (2023a).
\newblock The geometry of constant function market makers.
\newblock {\em arXiv preprint arXiv:2308.08066}.

\bibitem[Angeris et~al., 2023b]{angeris2023replicating}
Angeris, G., Evans, A., and Chitra, T. (2023b).
\newblock Replicating market makers.
\newblock {\em Digital Finance}, 5(2):367--387.

\bibitem[Bichuch and Feinstein, 2024]{bichuch2024defi}
Bichuch, M. and Feinstein, Z. (2024).
\newblock Defi arbitrage in hedged liquidity tokens.
\newblock {\em arXiv preprint arXiv:2409.11339}.

\bibitem[Bichuch and Feinstein, 2025]{bichuch2022axioms}
Bichuch, M. and Feinstein, Z. (2025).
\newblock Axioms for automated market makers: A mathematical framework in
  fintech and decentralized finance.
\newblock {\em Operations Research}.

\bibitem[Boyd and Vandenberghe, 2004]{boyd2004convex}
Boyd, S. and Vandenberghe, L. (2004).
\newblock {\em Convex optimization}.
\newblock Cambridge university press.

\bibitem[Chen and Pennock, 2012]{chen2012utility}
Chen, Y. and Pennock, D.~M. (2012).
\newblock A utility framework for bounded-loss market makers.
\newblock {\em arXiv preprint arXiv:1206.5252}.

\bibitem[Ciurlia and Caperdoni, 2009]{ciurlia2009note}
Ciurlia, P. and Caperdoni, C. (2009).
\newblock A note on the pricing of perpetual continuous-installment options.
\newblock {\em Mathematical Methods in Economics and Finance}, 4(1):11--26.

\bibitem[Diop et~al., 2024]{diop2024collapse}
Diop, P.~O., Chevallier, J., and Sanhaji, B. (2024).
\newblock Collapse of silicon valley bank and usdc depegging: A machine
  learning experiment.
\newblock {\em FinTech}, 3(4):569--590.

\bibitem[Duffie and Zhu, 2011]{duffie2011does}
Duffie, D. and Zhu, H. (2011).
\newblock Does a central clearing counterparty reduce counterparty risk?
\newblock {\em The Review of Asset Pricing Studies}, 1(1):74--95.

\bibitem[Eskandari et~al., 2021]{eskandari2021sok}
Eskandari, S., Salehi, M., Gu, W.~C., and Clark, J. (2021).
\newblock Sok: Oracles from the ground truth to market manipulation.
\newblock In {\em Proceedings of the 3rd ACM Conference on Advances in
  Financial Technologies}, pages 127--141.

\bibitem[Fateh~Singh et~al., 2025]{fateh2025option}
Fateh~Singh, S., Nekriach, V., Michalopoulos, P., Veneris, A., and Klinck, J.
  (2025).
\newblock Option contracts in the defi ecosystem: Opportunities, solutions, and
  technical challenges.
\newblock {\em International Journal of Network Management}, 35(2):e70005.

\bibitem[Feinstein, 2026]{feinstein2025amortizing}
Feinstein, Z. (2026).
\newblock Amortizing perpetual options.
\newblock {\em arXiv preprint arXiv:2512.06505}.

\bibitem[Gogol et~al., 2024]{gogol2024sok}
Gogol, K., Velner, Y., Kraner, B., and Tessone, C. (2024).
\newblock Sok: Liquid staking tokens (lsts) and emerging trends in restaking.
\newblock {\em arXiv preprint arXiv:2404.00644}.

\bibitem[Gudgeon et~al., 2020]{gudgeon2020defi}
Gudgeon, L., Werner, S., Perez, D., and Knottenbelt, W.~J. (2020).
\newblock Defi protocols for loanable funds: Interest rates, liquidity and
  market efficiency.
\newblock In {\em Proceedings of the 2nd ACM Conference on Advances in
  Financial Technologies}, pages 92--112.

\bibitem[Hanson, 2007]{hanson2007logarithmic}
Hanson, R. (2007).
\newblock Logarithmic market scoring rules for modular combinatorial
  information aggregation.
\newblock {\em The Journal of Prediction Markets}, 1(1):3--15.

\bibitem[Hull, 2018]{hull}
Hull, J.~C. (2018).
\newblock {\em Options, Futures, and Other Derivatives}.
\newblock Pearson, 10th edition.

\bibitem[Kimura, 2009]{kimura2009american}
Kimura, T. (2009).
\newblock American continuous-installment options: Valuation and premium
  decomposition.
\newblock {\em SIAM Journal on Applied Mathematics}, 70(3):803--824.

\bibitem[Kimura, 2010]{kimura2010valuing}
Kimura, T. (2010).
\newblock Valuing continuous-installment options.
\newblock {\em European Journal of Operational Research}, 201(1):222--230.

\bibitem[Kozhan and Viswanath-Natraj, 2021]{kozhan2021decentralized}
Kozhan, R. and Viswanath-Natraj, G. (2021).
\newblock Decentralized stablecoins and collateral risk.
\newblock {\em WBS Finance Group Research Paper Forthcoming}, pages 1--28.

\bibitem[Lambert and Kristensen, 2022]{lambert2022panoptic}
Lambert, G. and Kristensen, J. (2022).
\newblock Panoptic: the perpetual, oracle-free options protocol.
\newblock {\em arXiv preprint arXiv:2204.14232}.

\bibitem[Leshner and Hayes, 2019]{leshner2019compound}
Leshner, R. and Hayes, G. (2019).
\newblock Compound: The money market protocol.
\newblock {\em White Paper}, 93.

\bibitem[Milionis et~al., 2022]{milionis2022automated}
Milionis, J., Moallemi, C.~C., Roughgarden, T., and Zhang, A.~L. (2022).
\newblock Automated market making and loss-versus-rebalancing.
\newblock {\em arXiv preprint arXiv:2208.06046}.

\bibitem[{Nexus Mutual}, 2025]{nexus}
{Nexus Mutual} (2025).
\newblock Don't worry about depegs: Introducing depeg cover from nexus mutual.
\newblock
  \url{https://nexusmutual.io/blog/dont-worry-about-depegs-introducing-depeg-cover-from-nexus-mutual}.

\bibitem[{OECD}, 2022]{oecd2022lessons}
{OECD} (2022).
\newblock Lessons from the crypto winter: {DeFi} versus {CeFi}.
\newblock Technical report, {OECD} Business and Finance Policy Papers.

\bibitem[Pennella et~al., 2025]{pennella2025unified}
Pennella, L., Saggese, P., Pinelli, F., and Galletta, L. (2025).
\newblock A unified framework and comparative study of decentralized finance
  derivatives protocols.
\newblock {\em arXiv preprint arXiv:2512.19113}.

\bibitem[Pirrong, 2011]{pirrong2011economics}
Pirrong, C. (2011).
\newblock The economics of central clearing: theory and practice.

\bibitem[Qin et~al., 2021]{qin2021attacking}
Qin, K., Zhou, L., Livshits, B., and Gervais, A. (2021).
\newblock Attacking the defi ecosystem with flash loans for fun and profit.
\newblock In {\em International conference on financial cryptography and data
  security}, pages 3--32. Springer.

\bibitem[Schlegel et~al., 2022]{schlegel2022axioms}
Schlegel, J.~C., Kwa{\'s}nicki, M., and Mamageishvili, A. (2022).
\newblock Axioms for constant function market makers.
\newblock {\em arXiv preprint arXiv:2210.00048}.

\bibitem[Xu et~al., 2023]{xu2023sok}
Xu, J., Paruch, K., Cousaert, S., and Feng, Y. (2023).
\newblock Sok: Decentralized exchanges (dex) with automated market maker (amm)
  protocols.
\newblock {\em ACM Computing Surveys}, 55(11):1--50.

\bibitem[Zargham et~al., 2020]{zargham2020curved}
Zargham, M., Shorish, J., and Paruch, K. (2020).
\newblock From curved bonding to configuration spaces.
\newblock In {\em 2020 IEEE International Conference on Blockchain and
  Cryptocurrency (ICBC)}, pages 1--3. IEEE.

\end{thebibliography}

\end{document}